\documentclass[12pt]{article}
 \font\smallit=cmti10
 \font\smalltt=cmtt10
 \font\smallrm=cmr9 

\usepackage{amsmath,amsthm,amssymb,amsfonts,graphicx}

\theoremstyle{plain} 
\newtheorem{theorem}{Theorem}
\newtheorem{corollary}{Corollary}
\newtheorem{lemma}{Lemma}
\newtheorem{proposition}{Proposition}
\theoremstyle{definition}
\newtheorem{example}{Example}
\newtheorem{definition}{Definition}
\newtheorem{Alg}{Algorithm}
\newtheorem{remark}{Remark}

\def\val#1{\mathrm{val}_{#1}}
\def\rep#1{\mathrm{rep}_{#1}}

\begin{document}

\begin{center}
 {\bf ABSTRACT NUMERATION SYSTEMS ON BOUNDED LANGUAGES AND MULTIPLICATION BY A CONSTANT}
 \vskip 20pt
 {\bf Emilie Charlier}\\
 {\smallit Institute of Mathematics, University of Li\`ege, Grande Traverse 12 (B 37),
  B--4000 Li\`ege, Belgium}\\
 {\tt echarlier@ulg.ac.be}\\ 
 \vskip 10pt
 {\bf Michel Rigo}\footnote{The first two authors were supported by an AutoMathA short visit grant (European Science Fundation).}\\
 {\smallit Institute of Mathematics, University of Li\`ege, Grande Traverse 12 (B 37),
  B--4000 Li\`ege, Belgium}\\
 {\tt M.Rigo@ulg.ac.be}\\ 
 \vskip 10pt
 {\bf Wolfgang Steiner\footnote{This author was supported by the French Agence Nationale de la Recherche, grant ANR--06--JCJC--0073.}}\\
 {\smallit LIAFA, CNRS, Universit\'e Paris Diderot -- Paris 7, case 7014, 75205 Paris Cedex 13, France}\\
 {\tt steiner@liafa.jussieu.fr}\\ 
 \end{center}
 \vskip 30pt
 \centerline{\smallit Received: , Accepted: , Published: } 
 \vskip 30pt 

\centerline{\bf Abstract}
 A set of integers is $S$-recognizable in an abstract numeration 
    system $S$ if the language made up of the representations of its 
    elements is accepted by a finite automaton.
    For abstract numeration systems built over bounded languages with 
    at least three letters, we show that multiplication by an integer 
    $\lambda\ge2$ does not preserve $S$-recognizability, meaning that
    there always exists a $S$-recognizable set $X$ such that 
    $\lambda X$ is not $S$-recognizable. 
    The main tool is a bijection between the representation of an 
    integer over a bounded language and its decomposition as a sum of 
    binomial coefficients with certain properties, the so-called 
    combinatorial numeration system. 

\noindent

\pagestyle{myheadings}
 \markright{\smalltt INTEGERS: \smallrm ELECTRONIC JOURNAL OF COMBINATORIAL NUMBER THEORY \smalltt x (200x), \#Axx\hfill} 

 \thispagestyle{empty} 
 \baselineskip=15pt 
 \vskip 30pt

\section*{\normalsize 1. Introduction}

An {\em alphabet} is a finite set whose elements are called {\em
  letters}.  For a given alphabet $\Sigma$, a {\em word} of length
$n\ge 0$ over $\Sigma$ is a map $w:\{1,\ldots,n\}\to\Sigma$. The
length of a word $w$ is denoted by $|w|$.  The only word of length $0$
is the {\em empty word} denoted by $\varepsilon$.  The set of all
words over $\Sigma$ is $\Sigma^*$. The {\em concatenation} of the
words $u$ and $v$ respectively of length $m$ and $n$ is the word
$w=uv$ of length $m+n$ where $w(i)=u(i)$ for $1\le i\le m$ and
$w(i)=v(i-m)$ for $m+1\le i\le m+n$. Endowed with the concatenation
product, $\Sigma^*$ is a monoid with $\varepsilon$ as identity
element. For a word $u$ and $j\in\mathbb{N}$, $u^j$ is the
concatenation of $j$ copies of $u$. In particular, we set
$u^0=\varepsilon$. We write
$\Sigma^+=\Sigma^*\setminus\{\varepsilon\}$. A {\em language} over
$\Sigma$ is a subset of $\Sigma^*$. Since we use $|\cdot|$ to
denote the length of a word, we have chosen to denote the cardinality
of the set $A$ by $\#A$ to avoid any misunderstanding.

Denote the {\em bounded language} over the alphabet $\Sigma_\ell=
\{a_1,a_2,\ldots,a_\ell\}$ of size $\ell\ge 1$ by
$$\mathcal{B}_\ell=a_1^*a_2^*\cdots
a_\ell^*:=\{a_1^{j_1}a_2^{j_2}\cdots a_\ell^{j_\ell}\mid
j_1,j_2,\ldots,j_\ell\ge 0\}.$$
We always assume that $(\Sigma_\ell,<)$ is
totally ordered by $a_1<a_2<\cdots <a_\ell$. Let $x,y\in\Sigma_\ell^*$
be two words.  Recall that $x$ is {\em genealogically less than} $y$
either if $|x|<|y|$ or if they have the same length and $x$ is
lexicographically smaller than $y$, i.e., there exist
$p,x',y'\in\Sigma_\ell^*$ such that $x=pa_ix'$, $y=pa_jy'$ and $i<j$.
We can enumerate the words of $\mathcal{B}_\ell$ using the increasing
genealogical ordering (also called radix order or shortlex order) induced by the
ordering $<$ of $\Sigma_\ell$.  For an integer $n\ge 0$, the
$(n+1)$-st word of $\mathcal{B}_\ell$ is said to be the {\em
  $\mathcal{B}_\ell$-representation} of $n$ and is denoted by
$\rep{\ell}(n)$.  The reciprocal map $\rep{\ell}^{-1}=:\val{\ell}$
maps the $n$-th word of $\mathcal{B}_\ell$ onto its {\em numerical
  value} $n-1$. Notice that this map $\val{\ell}$ is a special case of
a diagonal function as considered for instance in \cite{Lew}.  A set
$X\subseteq\mathbb{N}$ is said to be {\em
  $\mathcal{B}_\ell$-recognizable} if $\rep{\ell}(X)$ is a regular
language over the alphabet $\Sigma_\ell$, i.e., accepted by a finite
automaton. This one-to-one correspondence between the words of
$\mathcal{B}_\ell$ and the integers can be extended to any infinite
regular language $L$ over a totally ordered alphabet $(\Sigma,<)$.
This leads to the general notion of abstract numeration system.
\begin{definition}
    An {\em abstract numeration system} is a triple $S=(L,\Sigma,<)$
    where $L$ is an infinite regular language over the totally ordered
    alphabet $(\Sigma,<)$. We denote by $\rep{S}(n)$ the $(n+1)$-st
    word in the genealogically ordered language $L$. A set $X$ of
    integers is {\em $S$-recognizable} if $\rep{S}(X)$ is a regular
    language.
\end{definition}

For an abstract numeration system $S=(L,\Sigma,<)$ where
$L=\mathcal{B}_\ell$ and $\Sigma=\Sigma_\ell$, the map $\rep{S}$ is
exactly $\rep{\ell}$.  Thus $\mathcal{B}_\ell$-recognizability is a
special case of $S$-recognizability.

Note that the language $\mathcal{B}_\ell$ is recognized by the following
automaton: the set of states is $\{q_1,\ldots,q_\ell\}$, each state is
final, $q_1$ is initial, and for $1\le i\le j\le n$ we have a
transition $q_i\xrightarrow{a_j}q_j$.
The case $\ell=4$ is depicted in Figure~\ref{fig:aut}.

\begin{figure}[htbp]
\centering\includegraphics{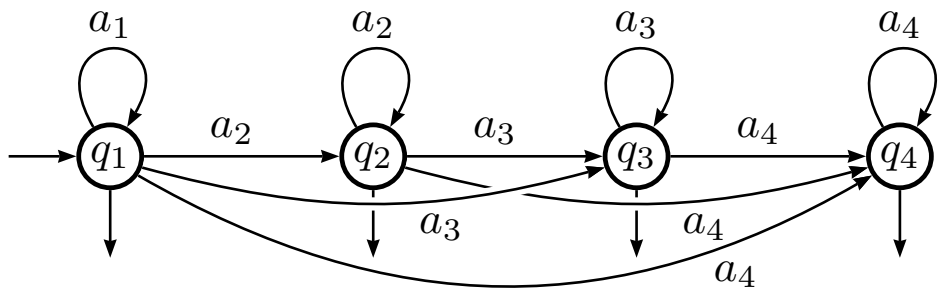}
\caption{Automaton recognizing $\mathcal{B}_4$.}
\label{fig:aut}
\end{figure}

\begin{example}
     Let $\Sigma_2=\{a,b\}$ with $a<b$.  The first words of
  $\mathcal{B}_2=a^*b^*$ enumerated by genealogical order are
  $$\varepsilon,a,b,aa,ab,bb,aaa,aab,abb,bbb,aaaa,\ldots$$
  For instance,
  $\rep{2}(5)=bb$ and $\val{2}(a^*)=\{0,1,3,6,10,\ldots\}$ is a
  $\mathcal{B}_2$-recognizable subset of $\mathbb{N}$ (formed of all
  triangular numbers).
\end{example}

For details on bounded languages, see for instance~\cite{GS} and for a
reference on automata and formal languages theory, see~\cite{Ei}. 

In the framework of positional numeration systems, recognizable sets
of integers have been extensively studied since the seminal work of
A.~Cobham in the late sixties (see for instance \cite[Chap.  V]{Ei}).
Since then, the notion of recognizability has been studied from various
points of view (logical characterization, automatic sequences,
\ldots). In particular, recognizability for generalized number systems
like the Fibonacci system has been considered \cite{BH,Sh}. Here we
shall consider recognizable sets of integers in the general setting of
{\em abstract numeration systems}. It is well-known that the class of
regular languages $L$ splits into two parts with respect to the
behavior of the function $n\mapsto\#(L\cap \Sigma^n)$ \cite{Sz}. This
latter function is either bounded from above by $n^k$ for some $k$ or,
infinitely often bounded from below by $\theta^n$ for some $\theta>1$.
In these cases, we speak respectively of {\em polynomial} and {\em
  exponential languages}.

Notice that usual positional numeration systems like integer base
systems or the Fibonacci system are special cases of abstract 
numeration systems built on an exponential language. 
On the other hand, bounded languages are polynomial and this leads to 
new phenomena.

The question addressed in the present paper deals with the
preservation of the recognizability with respect to the operation of
multiplication by a constant. {\em Let $S=(L,\Sigma,<)$ be an abstract
  numeration system, $X$ be a $S$-recognizable set of integers and
  $\lambda$ be a positive integer. What can be said about the
  $S$-recognizability of $\lambda X$ ?} This question is a first step
before handling more complex operations such as addition of two
arbitrary recognizable sets.

This question is rather difficult. For exponential languages, partial
answers are known (see for instance \cite{BH}). The case of polynomial
languages has not been considered yet (except for $a^*b^*$ in
\cite{LR}). Bounded languages are good candidates to start with.
Indeed, an arbitrary polynomial language is a finite union of
languages of the form $u_1v_1^*u_2v_2^*\cdots v_k^*u_{k+1}$ where the
$u_i$'s and $v_i$'s are words \cite{Sz}, and the automata accepting these
languages share the same properties as those accepting bounded
languages. Therefore we hope that our results give the flavor of what
could be expected for any polynomial languages.

Since $\rep{\ell}$ is a one-to-one correspondence between $\mathbb{N}$
and $\mathcal{B}_\ell$, the multiplication by a constant
$\lambda\in\mathbb{N}$ can be viewed as a transformation
$f_\lambda:\mathcal{B}_\ell\to\mathcal{B}_\ell$ acting on the language
$\mathcal{B}_\ell$, the question being then to {\em study the
preservation of the regularity of the subsets of $\mathcal{B}_\ell$
under this transformation}. 

\begin{example}
Let $\ell=2$, $\Sigma_2=\{a,b\}$ and $\lambda=25$. 
We have the following diagram.
$$ 
\begin{array}{rcl}
8 & \xrightarrow{\times 25}& 200\\
\rep{2}\downarrow & &\downarrow\rep{2} \\
a\, b^{2} & \xrightarrow{\times 25} & a^9\, b^{10}\\
\end{array}
\quad\quad
\begin{array}{rcl}
\mathbb{N} & \xrightarrow{\times \lambda}& \mathbb{N}\\
\rep{2}\downarrow & &\downarrow\rep{2} \\
\mathcal{B}_\ell & \xrightarrow{f_\lambda} & \mathcal{B}_\ell \\
\end{array}$$
Thus the multiplication by $\lambda=25$ induces a mapping $f_\lambda$ onto
$\mathcal{B}_2$ such that for $w,w'\in\mathcal{B}_2$, $f_\lambda(w)=w'$
if and only if $\val{2}(w')=25\, \val{2}(w)$.
\end{example}

This paper is organized as follows. In Section~2, we recall
a few results related to our main question. In particular, we
characterize the recognizable sets of integers for abstract numeration
systems whose language is slender, i.e., has at most $d$ words of each
length for some constant $d$. We easily get that in this situation, the
multiplication by a constant always preserves recognizability.

In Section~3, we
compute $\val{\ell}(a_1^{n_1}\cdots a_\ell^{n_\ell})$ and derive an
easy bijective proof of the fact that any nonnegative integer can be
written in a unique way as
$$n=\binom{z_\ell}{\ell}+\binom{z_{\ell-1}}{\ell-1}+\cdots+\binom{z_1}{1}$$
with $z_\ell>z_{\ell-1}>\cdots >z_1\ge 0$.  Fraenkel~\cite{Fra} called
this system {\em combinatorial numeration system} and referred to
Lehmer~\cite{Leh}.  Even if this seems to be a folklore result, the
only proof that we were able to trace out goes back to
Katona~\cite{Kat} who developed different arguments to obtain the same
decomposition.

In Section~4, we make explicit the regular subsets of
$\mathcal{B}_\ell$ in terms of semi-linear sets of $\mathbb{N}^\ell$
and give an application to the $\mathcal{B}_\ell$-recognizability of
arithmetic progressions. 

In Section~5, we answer our main question about bounded
languages and recognizability after multiplication by a constant. We
get a formula which can be used to obtain estimates on the
$\mathcal{B}_\ell$-representation of $\lambda n$ from the one of $n$.
Therefore, thanks to a counting argument and to the results from
Section~4, we show that for any constant $\lambda$, there
exists a $\mathcal{B}_\ell$-recognizable set $X$ such that $\lambda X$
is no more $\mathcal{B}_\ell$-recognizable, with $\ell\ge 3$.
Consequently, our main result can be summarized as follows. Let
$\ell,\,\lambda$ be positive integers. For the abstract
numeration system $S=(a_1^*\cdots a_\ell^*,\{a_1<\cdots<a_\ell\})$,
multiplication by $\lambda\ge2$ preserves $S$-recognizability if and
only if either $\ell=1$ or $\ell=2$ and $\lambda$ is an odd square.
    
    We put in the last section some structural results concerning the
    effect of multiplication by a constant in the abstract numeration
    system built on $\mathcal{B}_\ell$.

\vskip 30pt

\section*{\normalsize 2. First results about $S$-recognizability}

In this section we collect a few results directly connected with our
problem.

\begin{theorem}\cite{LR}\label{the:pa}
    Let $S=(L,\Sigma,<)$ be an abstract numeration system. Any
    arithmetic progression is $S$-recognizable.
\end{theorem}

Let us denote by $\mathbf{u}_L(n)$ (resp. $\mathbf{v}_L(n)$) the number of words of length $n$ (resp. at most $n$) 
belonging to $L$. The following result states that only some constants 
$\lambda$ are good candidates for multiplication within
$\mathcal{B}_\ell$.

\begin{theorem}\cite{Ri}\label{the:6}
    Let $L\subseteq\Sigma^*$ be a regular language such that
    $\mathbf{u}_L(n)=\Theta(n^k)$ for some $k\in\mathbb{N}$ and
    $S=(L,\Sigma,<)$. Preservation of $S$-recognizability after
    multiplication by $\lambda$ holds only if $\lambda=\beta^{k+1}$
    for some $\beta\in\mathbb{N}$.
\end{theorem}

We write $f=\Theta(g)$ if there exist $N$ and $C>0$ such that for 
all $n\ge N$, $f(n)\le C\, g(n)$ (i.e., $f=\mathcal{O}(g)$) and also if 
there exist $D>0$ and an infinite sequence $(n_i)_{i\in\mathbb{N}}$ 
such that $f(n_i)\ge D\, g(n_i)$ for all $i\ge 0$.

As we shall see in the next section that
$\mathbf{u}_{\mathcal{B}_\ell}(n)=\Theta(n^{\ell-1})$, we have to
focus only on multipliers of the form $\beta^\ell$. 
The particular case of $\mathbf{u}_L(n)=\mathcal{O}(1)$ (i.e., $L$ is 
{\em slender}) is interesting in itself and is settled as follows. 
Let us first recall the definition from~\cite{ADPS} and the 
characterization from~\cite{PS,Sh} of such languages.
  
\begin{definition}
    The language $L$ is said to be {\em $d$-slender} if
    $\mathbf{u}_L(n)\le d$ for all $n\ge 0$. 
    The language $L$ is said to be {\em slender} if it is $d$-slender 
    for some $d>0$. 
\end{definition}

A regular language $L$ is slender if and only if it is a {\em union of 
single loops}, i.e., if for some $k\ge 1$ and words $x_i$, $y_i$, 
$z_i$, $1\le i\le k$,
$$
L={\displaystyle \bigcup_{i=1}^k}\, x_i\, y_i^* z_i.
$$
Moreover, we can assume that the sets $x_i\, y_i^* z_i$ are
pairwise disjoint. Notice that the regular expression $x_i\, y_i^*
z_i$ is a shorthand to denote the language $\{x_i y_i^n z_i\mid n\ge
0\}$, again $x_i y_i^n z_i$ has to be understood as the concatenation
of $x_i$, $n$ copies of $y_i$ and then followed by $z_i$.

\begin{theorem}\label{the:slender}
    Let $L\subseteq \Sigma^*$ be a slender regular language and
    $S=(L,\Sigma,<)$. A~set $X\subseteq \mathbb{N}$ is $S$-recognizable
    if and only if $X$ is a finite union of arithmetic progressions.
\end{theorem}

\begin{proof}
    By the characterization of slender languages, we have
$$L={\displaystyle \bigcup_{i=1}^k}\, x_i\, y_i^* z_i \cup F,\ 
    x_i,z_i \in\Sigma^*, y_i\in\Sigma^+$$
    where the sets $x_i\, y_i^*
    z_i$ are pairwise disjoint and $F$ is a finite set. The sequence
    $(\mathbf{u}_L(n))_{n\in\mathbb{N}}$ is ultimately periodic of
    period $C=\mathrm{lcm}_i |y_i|$. 
    Moreover, for $n$ large enough, if
    $x_i\, y_i^n\, z_i$ is the $m$-th word of length $|x_i\, z_i|+n\,
    |y_i|$ then $x_i\, y_i^{n+C/|y_i|}\, z_i$ is the $m$-th word
    of length $|x_i\, z_i|+n\, |y_i|+C$. Roughly speaking, for
    sufficiently large $n$, the structures of the ordered sets of words 
    of length $n$ and $n+C$ are the same.\\
    The regular subsets of $L$ are of the form
    \begin{equation}\label{eq:slender2}
    \bigcup_{j\in J}\, x_{i_j}\, (y_{i_j}^{\alpha_j})^* z_{i_j} \cup F'
    \end{equation}
where $J$ is a finite set, $i_j\in\{1,\ldots,k\}$, 
$\alpha_j\in \mathbb{N}$ and $F'$ is a finite subset of $L$. \\
We can now conclude. If $X$ is $S$-recognizable, then $\rep{S}(X)$ is
a regular subset of $L$ of the form \eqref{eq:slender2}. In view of
the first part of the proof, it is clear that $X$ is ultimately
periodic with period length $\mathrm{lcm}(C,\mathrm{lcm}_j|y_{i_j}\alpha_j|)$. 
The converse is immediate by Theorem \ref{the:pa}. 
\end{proof}

\begin{example}
    Consider the language $L=ab^*c\cup b(aa)^*c$. It contains exactly
    two words of each positive even length: $ab^{2i}c<ba^{2i}c$ and
    one word for each odd length larger than $2$: $ab^{2i+1}c$. The
    sequence $\mathbf{u}_L(n)$ is ultimately periodic of period two:
    $0,0,2,1,2,1,\ldots$.
\end{example}

\begin{corollary} 
    Let $S$ be a numeration system built on a slender language. If
    $X\subseteq \mathbb{N}$ is $S$-recognizable, then $\lambda X$ is
    $S$-recognizable for all $\lambda \in \mathbb{N}$.
\end{corollary}

Finally, for a bounded language over a binary alphabet, the case is
completely settled too, the aim of this paper being primarily to
extend the following result.

\begin{theorem}\cite{LR}\label{the:LR}
    Let $\beta$ be a positive integer.  
    For the abstract numeration system $S=(a^*b^*,\{a<b\})$, 
    multiplication by $\beta^2$ preserves $S$-recognizability if and 
    only if $\beta$ is odd.
\end{theorem}

 \vskip 30pt

\section*{\normalsize 3. $\mathcal{B}_\ell$-representation of integers : combinatorial 
expansion}

In this section we determine the number of words of a given length in
$\mathcal{B}_\ell$ and we obtain an algorithm for computing
$\rep{\ell}(n)$. Interestingly, this algorithm is related to the
decomposition of $n$ as a sum of binomial coefficients of a specified
form. 
Since we shall be mainly interested by the language $\mathcal{B}_\ell$, 
we use the following notation.

\begin{definition}\label{def:not}
We set
$$
\mathbf{u}_\ell(n):=\mathbf{u}_{\mathcal{B}_\ell}(n)
=\#(\mathcal{B}_\ell\cap\Sigma_\ell^n)\quad \text{ and }\quad
\mathbf{v}_\ell(n):=\#(\mathcal{B}_\ell\cap\Sigma_\ell^{\le
  n})=\sum_{i=0}^n\mathbf{u}_\ell(i).
$$
\end{definition}

Let us also recall that the
binomial coefficient $\binom{i}{j}$ vanishes for integers $i<j$.

\begin{lemma}\label{lem:11}
    For all $\ell\ge 1$ and $n\ge 0$, we have 
    \begin{equation}\label{eq:uv}
\mathbf{u}_{\ell+1}(n)=\mathbf{v}_\ell(n) 
    \end{equation}
and 
    \begin{equation}\label{eq:ubin}
\mathbf{u}_\ell(n)=\binom{n+\ell-1}{\ell-1}.
  \end{equation}
\end{lemma}

\begin{proof}
  Relation \eqref{eq:uv} follows from the fact that the set of words
  of length $n$ belonging to $\mathcal{B}_{\ell+1}$ is partitioned
  according to
  $$\bigcup_{i=0}^n \left(a_1^*\cdots a_\ell^*\cap
      \Sigma_\ell^{i}\right)a_{\ell+1}^{n-i}.$$
    To obtain \eqref{eq:ubin}, we proceed by induction on $\ell\ge 1$.
    Indeed, for $\ell=1$, it is clear that $\mathbf{u}_1(n)=1$ for all
    $n\ge 0$. Assume that \eqref{eq:ubin} holds for $\ell$ and let us
    verify it still holds for $\ell+1$. Thanks to \eqref{eq:uv}, we
    have
  $$
  \hskip19mm\mathbf{u}_{\ell+1}(n)=\sum_{i=0}^n
  \mathbf{u}_{\ell}(i)=\sum_{i=0}^n
  \binom{i+\ell-1}{\ell-1}=\sum_{i=0}^n \binom{i+\ell-1}{i}=
  \binom{n+\ell}{\ell}.\hskip14mm\qedhere
  $$
\end{proof}

\begin{lemma}\label{lem:val}
    Let $S=(a_1^*\cdots a_\ell^*,\{a_1<\cdots <a_\ell\})$. We have
    \begin{equation}\label{eq:vall}
    \val{\ell}(a_1^{n_1}\cdots a_\ell^{n_\ell})=\sum_{i=1}^\ell
    \binom{n_i+\cdots +n_\ell+\ell-i}{\ell-i+1}.
    \end{equation}
    Consequently, for any $n\in\mathbb{N}$,
    $$
    |\rep{\ell}(n)|=k\Leftrightarrow \underbrace{\binom{k+\ell-1}{\ell}}_{\val{\ell}(a_1^k)} \le n\le \underbrace{\sum_{i=1}^\ell
    \binom{k+i-1}{i}}_{\val{\ell}(a_\ell^k)}.
    $$
\end{lemma}

\begin{proof}
    From the structure of the ordered language $\mathcal{B}_\ell$, one
    can show that
    \begin{equation}
        \label{eq:decompp}
        \val{\ell}(a_1^{n_1}\cdots
    a_\ell^{n_\ell})=\val{\ell}(a_1^{n_1+\cdots
      +n_\ell})+\val{\{a_2,\ldots,a_\ell\}}(a_2^{n_2}\cdots
    a_{\ell}^{n_\ell})
    \end{equation}
    where notation like
    $\val{\{a_2,\ldots,a_\ell\}}(w)$ specifies not only the size but the
    alphabet of the bounded language on which the numeration system is
    built. To understand this formula, an example is given below in
    the case $\ell=3$. Notice that $\val{\{a_2,\ldots,a_\ell\}}(a_2^{n_2}\cdots
    a_{\ell}^{n_\ell})=\val{\ell-1}(a_1^{n_2}\cdots
    a_{\ell-1}^{n_\ell})$. Using this latter observation and iterating the decomposition \eqref{eq:decompp}, we obtain
    $$\val{\ell}(a_1^{n_1}\cdots a_\ell^{n_\ell})=\sum_{i=1}^{\ell}
    \val{\ell-i+1}(a_1^{n_i+\cdots +n_\ell}).$$
    Moreover, it is
    well known that $\val{\ell}(a_1^n)=\mathbf{v}_\ell(n-1)$.  Hence
    the conclusion follows using relations \eqref{eq:uv} and
    \eqref{eq:ubin}.
\end{proof}

\begin{example}
    Consider the words of length $3$ in the language $a^*b^*c^*$,
    $$aaa<aab<aac<abb<abc<acc<bbb<bbc<bcc<ccc.$$
    We have
    $\val{3}(aaa)=\binom{5}{3}=10$ and $\val{3}(acc)=15$. If we apply
  the erasing morphism $\varphi:\{a,b,c\}\to\{a,b,c\}^*$ defined by
  $\varphi(a)=\varepsilon$, $\varphi(b)=b$ and $\varphi(c)=c$ on the
  words of length $3$, we get
  $$\varepsilon <b<c<bb<bc<cc<bbb<bbc<bcc<ccc.$$
  So the ordered list
  of words of length $3$ in $a^*b^*c^*$ contains an ordered copy of
  the words of length at most $2$ in the language $b^*c^*$ and to
  obtain $\val{3}(acc)$, we just add to $\val{3}(aaa)$ the position of
  the word $cc$ in the ordered language $b^*c^*$. In other words,
  $\val{3}(acc)=\val{3}(aaa)+\val{2}(cc)$ where $\val{2}$ is
  considered as a map defined on the language $b^*c^*$. 
\end{example}

The following result is given in \cite{Kat}. Here we obtain a
bijective proof relying only on the use of abstract numeration systems
on a bounded language.

\begin{corollary}[Combinatorial numeration system]\label{cor:decomp}
    Let $\ell$ be a positive integer. Any integer $n\ge0$ can be
    uniquely written as
    \begin{equation}
        \label{eq:decompbin}
        n=\binom{z_\ell}{\ell}+\binom{z_{\ell-1}}{\ell-1}+\cdots+\binom{z_1}{1}
    \end{equation}
    with $z_\ell>z_{\ell-1}>\cdots >z_1\ge 0$.
\end{corollary}

\begin{proof}
    The map $\rep{\ell}:\mathbb{N}\to a_1^*\cdots a_\ell^*$ is a
    one-to-one correspondence. So any integer $n$ has a unique
    representation of the form $a_1^{n_1}\cdots a_\ell^{n_\ell}$ and
    the conclusion follows from Lemma~\ref{lem:val}.
\end{proof}

The general method given in \cite[Algorithm 1]{LR} has a special form
in the case of the language $\mathcal{B}_\ell$. We derive an algorithm
computing the decomposition \eqref{eq:decompbin} or equivalently the
$\mathcal{B}_\ell$-representation of any integer.

\begin{Alg}
    Let {\tt n} be an integer and {\tt l} be a positive integer. The following algorithm produces integers {\tt z(l)},\ldots,{\tt z(1)} corresponding to the $z_i$'s appearing in the decomposition \eqref{eq:decompbin} of {\tt n} given in Corollary \ref{cor:decomp}.\\

    \indent\indent {\tt For i=l,l-1,\ldots,1 do}\\
    \indent\indent\indent {\tt if n>0,}\\
    \indent\indent\indent\indent {\tt find t such that $\binom{{\tt t}}{{\tt i}}\le {\tt n}< \binom{{\tt t}+1}{{\tt i}}$}\\
    \indent\indent\indent\indent {\tt z(i)$\leftarrow$t}\\
    \indent\indent\indent\indent {\tt n$\leftarrow$n-$\binom{{\tt t}}{{\tt i}}$}\\
    \indent\indent\indent {\tt otherwise, z(i)$\leftarrow$i-1}\\

\noindent    Consider now the triangular system having
    $n_{1},\ldots,n_{\ell}$ as unknowns
    $$n_{i}+\cdots+n_{\ell}={\tt z}(\ell-i+1)-\ell+i,\quad i=1,\ldots
    ,\ell.$$
    One has $\rep{\ell}({\tt n})=a_1^{n_1}\cdots a_\ell^{n_\ell}$.
\end{Alg}

\begin{remark}
    To speed up the computation of {\tt t} in the above algorithm, one
    can benefit from methods of numerical analysis. Indeed, for given
    {\tt i} and {\tt n}, $\binom{{\tt t}}{{\tt i}}-{\tt n}$ is a
    polynomial in {\tt t} of degree {\tt i} and we are looking for the
    largest root $z$ of this polynomial. Therefore, ${\tt t}=\lfloor
    z\rfloor$.
\end{remark}

\begin{example}
    For $\ell=3$, one gets for instance
$$12345678901234567890=\binom{4199737}{3}+\binom{3803913}{2}+\binom{1580642}{1}$$
and solving the system
$$\left.\begin{array}{rcl}
n_1+n_2+n_3&=&4199737-2\\
n_2+n_3&=&3803913-1\\
n_3&=&1580642\\
\end{array}\right\} \Leftrightarrow (n_1,n_2,n_3)=(395823,2223270,1580642),$$
we have
$\rep{3}(12345678901234567890)=a^{395823}b^{2223270}c^{1580642}$.
\end{example}

 \vskip 30pt

\section*{\normalsize 4. Regular subsets of $\mathcal{B}_\ell$}

To study preservation of recognizability after multiplication by a
constant, one has to consider an arbitrary recognizable subset
$X\subseteq\mathbb{N}$ and show that $\beta^\ell X$ is still
recognizable. 

\begin{definition}
    If $w$ is a word over $\Sigma_\ell$, $|w|_{a_j}$ counts the number
    of letters $a_j$ in~$w$. The {\em Parikh mapping}
    $\Psi$ maps a word $w\in\Sigma_\ell^*$ onto the vector
    $\Psi(w):=(|w|_{a_1},\ldots,|w|_{a_\ell})$.
\end{definition}

\begin{remark}
      In this setting of
    bounded languages, $\rep{\ell}$ and $\Psi$ are both one-to-one
    correspondences. Therefore, in what follows we shall make no
    distinction between an integer $n$, its
    $\mathcal{B}_\ell$-representation $\rep{\ell}(n)=a_1^{n_1}\cdots
    a_\ell^{n_\ell}\in\mathcal{B}_\ell$ and the corresponding Parikh
    vector $\Psi(\rep{\ell}(n))=(n_1,\ldots,n_\ell)\in\mathbb{N}^\ell$. 
    In examples, when considering cases $\ell=2$ or $3$, we shall use
    convenient alphabets like $\{a<b\}$ or $\{a<b<c\}$.
\end{remark}

\begin{definition}
    A set $Z\subseteq\mathbb{N}^\ell$ is {\em linear} if there exist
    ${\mathbf p}_0,{\mathbf p}_1,\ldots,{\mathbf p}_k\in\mathbb{N}^\ell$ such that
    $$
    Z={\mathbf p}_0+\mathbb{N}\, {\mathbf p}_1+\cdots +\mathbb{N}\, {\mathbf p}_k=\{{\mathbf p}_0+\lambda_1 {\mathbf p}_1
    +\cdots+\lambda_k {\mathbf p}_k\mid \lambda_1,\ldots,\lambda_k\in\mathbb{N}\}.
    $$
    The vectors
    ${\mathbf p}_1,\ldots,{\mathbf p}_k$ are said to be the {\em periods} of $Z$.  The set
    $Z$ is {\em $k$-dimensional} if it has exactly $k$ linearly
    independent periods over $\mathbb{Q}$.  A set is {\em semi-linear}
    if it is a finite union of linear sets. The set of periods of a
    semi-linear set is the union of the sets of periods of the
    corresponding linear sets. 
    Let $\mathbf{e}_i\in\mathbb N^\ell$, $1\le i\le\ell$, denote the 
    vector having $1$ in the $i$-th component and $0$ in the other  
    components.
\end{definition}

\begin{lemma}\label{lem:rec}
    A set $X\subseteq\mathbb{N}$ is $\mathcal{B}_\ell$-recognizable if 
    and only if $\Psi(\rep{\ell}(X))$ is a semi-linear set whose 
    periods are integer multiples of canonical vectors $\mathbf{e}_i$.
\end{lemma}

\begin{proof}
Observe that the regular subsets of
$\mathcal{B}_\ell$ are exactly the finite unions of sets of the form
$a_1^{s_1}(a_1^{t_1})^*\cdots a_\ell^{s_\ell}(a_\ell^{t_\ell})^*$
with $s_i,t_i\in\mathbb N$. 
\end{proof}

With such a characterization, we obtain an
alternative proof of Theorem~\ref{the:pa}.

\begin{proposition}\label{pro:PA}
    Let $p,q\in\mathbb{N}$. 
    The set 
    $\Psi(\rep{\ell}(q+\mathbb{N}\,p))\subseteq\mathbb{N}^\ell$ is a 
    finite union of linear sets of the form
    $$
    \mathbf x+\mathbb{N}\,P\,\mathbf{e}_1+\cdots+
    \mathbb{N}\,P\,\mathbf{e}_\ell\quad\text{ for some }P\in\mathbb{N}.
    $$
\end{proposition}

\begin{proof}
    We use equation \eqref{eq:vall}. 
    For a given $i$, $1\le i\le\ell$, the sequence
    $(\binom n{\ell-i+1}\mod p)_{n\in\mathbb N}$ is periodic (see e.g.
    \cite{Za}).
    Denote the period lengths by $\pi_i$ and set
    $P=\mathrm{lcm}_i\,\pi_i$.
    Then
    $$
    \val{\ell}(a_1^{n_1}\cdots a_i^{n_i}\cdots a_\ell^{n_\ell})\equiv 
    \val{\ell}(a_1^{n_1}\cdots a_i^{n_i+P}\cdots a_\ell^{n_\ell}) 
    \pmod{p}\quad\mbox{ for all }i,\,1\le i\le\ell.
    $$
    We have just shown that 
    $\mathbf x=(x_1,\ldots,x_\ell)\in\mathbb{N}^\ell$ belongs to 
    $\Psi(\rep{\ell}(q+\mathbb{N}\, p))$ if and only if 
    $\mathbf x+n_1\,P\,\mathbf{e}_1+\cdots+n_\ell\,P\,\mathbf{e}_\ell$ 
    belongs to the same set for all $n_1,\ldots,n_\ell\in\mathbb{N}$. 
    Therefore
    $$
    \hskip28mm\Psi(\rep{\ell}(q+\mathbb{N}\, p))=
    \bigcup_{\substack{\val{\ell}(a_1^{x_1}\cdots a_\ell^{x_\ell})\in
        q+\mathbb{N}\, p\\ 0\le\sup x_i<q+P}}
    (\mathbf x+\mathbb{N}\,P\,\mathbf{e}_1+\cdots+
    \mathbb{N}\,P\,\mathbf{e}_\ell).\hskip23mm\qedhere
    $$
\end{proof}

\begin{example}
    In Figure \ref{fig:pa5}, the $x$-axis (resp.  $y$-axis) counts the
    number of $a_1$'s (resp. $a_2$'s) in a word.  The empty word
    corresponds to the lower-left corner. A point in $\mathbb{N}^2$ of
    coordinates $(i,j)$ has its color determined by the value of
    $\val{2}(a_1^i\, a_2^j)$ modulo $p$ (with $p=3,5,6$ and $8$
    respectively). There are therefore $p$ possible colors. In this
    figure, we represent words $a_1^i\, a_2^j$ for $0\le i,j\le 19$.
    \begin{figure}[htbp]
        \centering
        \includegraphics[width=2.8cm]{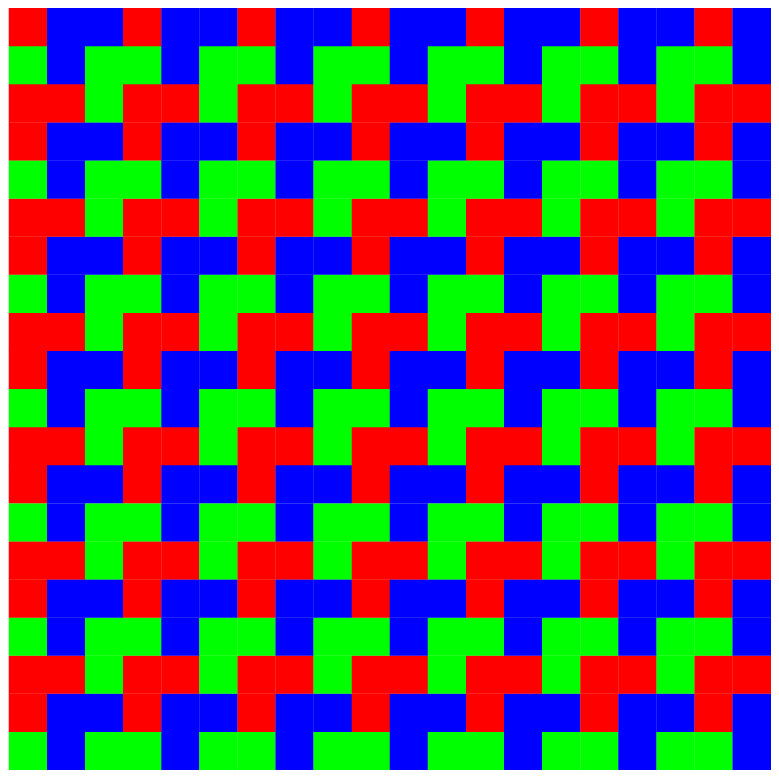}\quad 
        \includegraphics[width=2.8cm]{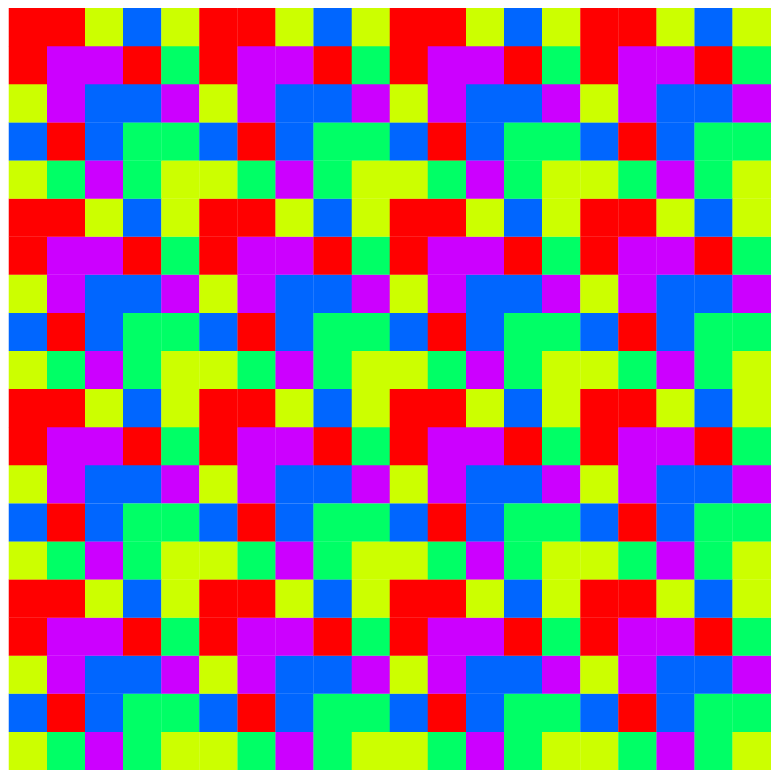}\quad 
        \includegraphics[width=2.8cm]{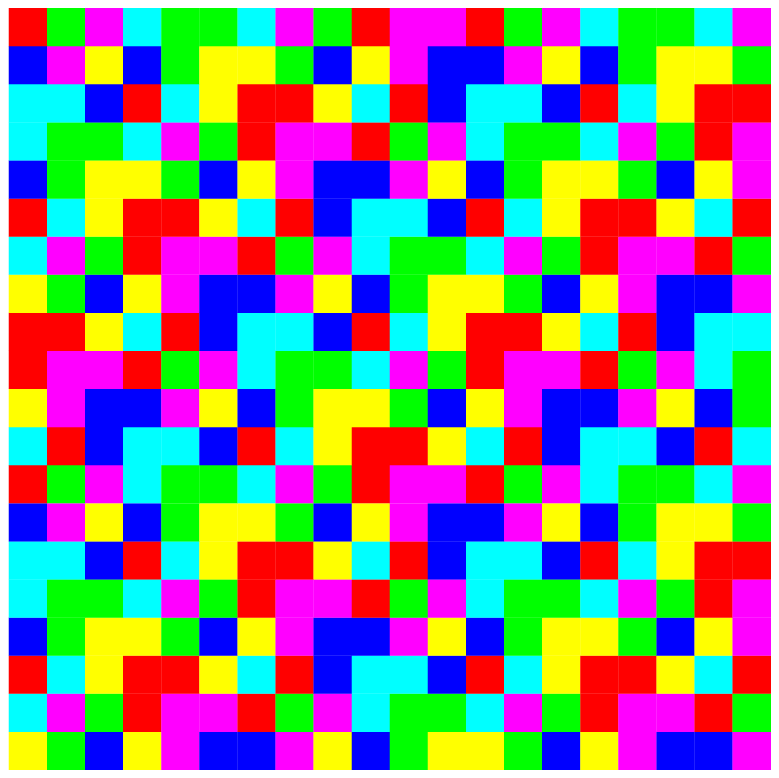}\quad 
        \includegraphics[width=2.8cm]{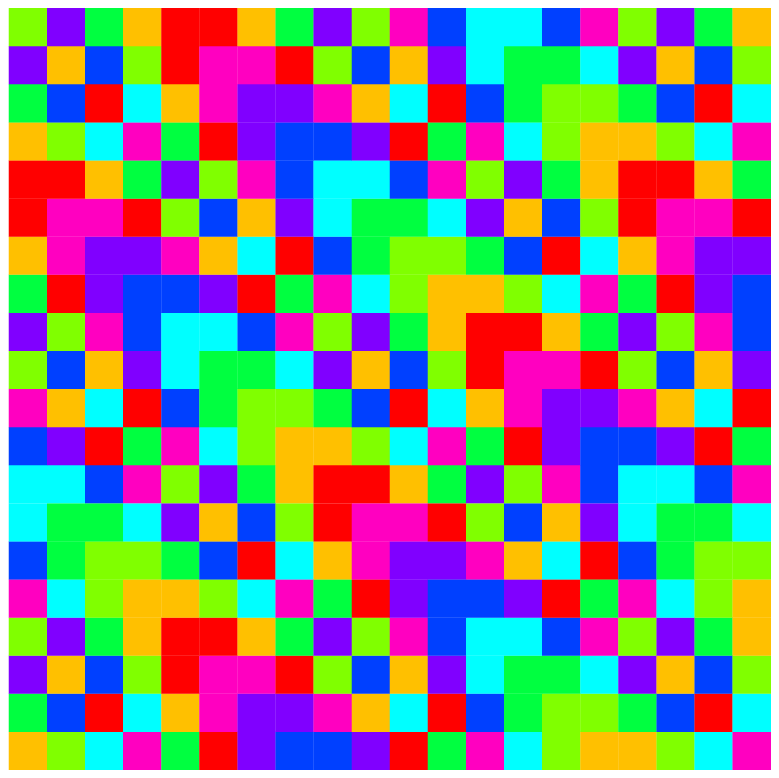}
        \caption{$\Psi(\rep{2}(q+\mathbb{N}\,p))$ for $p=3,5,6,8$.}
        \label{fig:pa5}
    \end{figure}   
\end{example}

 \vskip 30pt

\section*{\normalsize 5. Multiplication by $\lambda=\beta^\ell$}

In the case of a bounded language on $\ell$ letters, if multiplication
by some constant preserves recognizability, then, by Theorem
\ref{the:6} and Lemma \ref{lem:11}, this constant must be a $\ell$-th
power.

The next result gives a relationship between the length of the
$\mathcal{B}_\ell$-representations of $n$ and $\beta^\ell n$, roughly
by a factor $\beta$.

\begin{lemma}\label{lem:long}
For sufficiently large $n\in\mathbb N$, we have
$$
|\mathrm{rep}_\ell(\beta^\ell n)|=\beta\,|\mathrm{rep}_\ell(n)|+
\left\lceil\frac{(\beta-1)(\ell+1)}2\right\rceil-i
$$
for some $i\in\{0,1,\ldots,\beta\}$.
\end{lemma}

\begin{proof}
Consider first $n=\mathrm{val}_\ell(a_\ell^q)$ for some sufficiently
large $q\in\mathbb N$, and let 
\begin{multline*}
\beta^\ell\left(\binom{q+\ell-1}\ell+\binom{q+\ell-2}{\ell-1}+\cdots+\binom q1\right) \\ 
=\binom{z_\ell+\ell-1}{\ell}+\binom{z_{\ell-1}+\ell-2}{\ell-1}+\cdots+\binom{z_1}1
\end{multline*}
for some integers $z_\ell\ge z_{\ell-1}\ge\cdots\ge z_1\ge0$
(depending on $q$).  Then we have
$$
\beta^\ell\left(\frac{q^\ell}{\ell!}+\frac{(\ell+1)\,q^{\ell-1}}{2\,(\ell-1)!}+
\mathcal O(q^{\ell-2})\right)
=\frac{z_\ell^\ell}{\ell!}+\frac{(\ell-1)\,z_\ell^{\ell-1}}{2\,(\ell-1)!}+
\frac{z_{\ell-1}^{\ell-1}}{(\ell-1)!}+\mathcal O(z_\ell^{\ell-2}),
$$ 
thus $z_\ell=\beta q+\mathcal O(1)$.
Since $z_\ell\ge z_{\ell-1}$, we have $z_{\ell-1}=d\beta q+o(q)$ with 
$0\le d\le 1$ and we obtain
\begin{gather*}
\frac{\beta^\ell(\ell+1)}{2\,(\ell-1)!}q^{\ell-1}
=\frac{\beta^{\ell-1}}{(\ell-1)!}
\left((z_\ell-\beta q)+\frac{\ell-1}2+d^{\ell-1}\right)q^{\ell-1}
+\mathcal O(q^{\ell-2}), \\
z_\ell=\beta q+\frac{(\beta-1)(\ell+1)}2+1-d^{\ell-1}.
\end{gather*} 
Set $c=(\beta-1)(\ell+1)/2$ and assume first $c\not\in\mathbb Z$.
Then we have $d^{\ell-1}=1/2$, hence 
$$
|f_{\beta^\ell}(a_\ell^q))|=z_\ell=\beta q+\lceil c\rceil.
$$
Since $\mathrm{val}_\ell(a_1^q)=\mathrm{val}_\ell(a_\ell^{q-1})+1$, we have
$$
|\mathrm{rep}_\ell(\beta^\ell\mathrm{val}_\ell(a_1^q))|\ge
\beta(q-1)+\lceil c\rceil=\beta q+\lceil c\rceil-\beta.
$$
If $|\mathrm{rep}_\ell(n)|=q$, then $|\mathrm{rep}_\ell(\beta^\ell n)|$ is 
clearly between these two values. \\
Assume now $c\in\mathbb Z$.  Then we have $d\in\{0,1\}$.  Similarly to
the computation of $c_{\ell-2}$ achieved in Remark~\ref{rem:ccc}
below, we obtain that
\begin{align*}
\binom{\beta q+c+\ell}\ell & -\beta^\ell\binom{q+\ell}\ell \\
& = \left(\frac{c^2}2+\frac{(\ell+1)c}2+\frac{(1-\beta^2)(3\ell+2)(\ell+1)}{24}
\right)\frac{(\beta q)^{\ell-2}}{(\ell-2)!}+\mathcal O(q^{\ell-3}) \\ 
& = \frac{c(\beta+1)}{12}\frac{(\beta q)^{\ell-2}}{(\ell-2)!}+
\mathcal O(q^{\ell-3}). 
\end{align*}
This means that the numerical value of the first word of length 
$\beta q+c+1$ is larger than $\beta^\ell\mathrm{val}_\ell(a_1^{q+1})$ 
for large enough $q$.
We infer that $d=1$ since
$$
z_\ell=|\mathrm{rep}_\ell(\beta^\ell\mathrm{val}_\ell(a_\ell^q))| \le
|\mathrm{rep}_\ell(\beta^\ell\mathrm{val}_\ell(a_1^{q+1}))|<\beta q+c+1.
$$
As above, we have $|\mathrm{rep}_\ell(\beta^\ell\mathrm{val}_\ell(a_1^q))|\ge\beta q+c-\beta$, and
the lemma is proved.
\end{proof}

In certain cases, we can give a formula for the entire expansion of 
$\beta^\ell\mathrm{val}_\ell(a_\ell^q)$.

\begin{lemma}\label{lem:constants} 
Define $c_{\ell-1},c_{\ell-2},\ldots,c_0$ recursively by
$$
c_k =k!\,(\beta^{\ell-k}-1)\sum_{i=k}^\ell\frac{S_1(i,k)}{i!}-
\sum_{i=k+2}^\ell\sum_{j=k+1}^i\frac{S_1(i,j)\,j!}{i!\,(j-k)!}c_{i-1}^{j-k}
$$
where $S_1(i,j)$ are the unsigned Stirling numbers of the first kind. 
Then we have
\begin{multline}\label{eq:constants}
\beta^\ell\left(\binom{q+\ell-1}{\ell}+\binom{q+\ell-2}{\ell-1}+\cdots+
\binom q1\right) \\
=\binom{\beta q+c_{\ell-1}+\ell-1}{\ell}+
\binom{\beta q+c_{\ell-2}+\ell-2}{\ell-1}+\cdots+\binom{\beta q+c_0}1.
\end{multline}
Moreover, if all $c_k$'s, $0\le k<\ell$, are integers and 
$c_{\ell-1}\ge c_{\ell-2}\ge\cdots\ge c_0$, then
$$
\mathrm{rep}_\ell(\beta^\ell\mathrm{val}_\ell(a_\ell^q))=
a_1^{c_{\ell-1}-c_{\ell-2}}a_2^{c_{\ell-2}-c_{\ell-3}}\cdots 
a_{\ell-1}^{c_1-c_0}a_\ell^{\beta q+c_0}
$$
for all $q\ge-c_0/\beta$, hence 
$\mathrm{rep}_\ell(\beta^\ell\mathrm{val}_\ell(a_\ell^*))$ is regular.
\end{lemma}

\begin{proof}
The second part of the lemma is obvious.
Thus we only have to show (\ref{eq:constants}).
Recall that the unsigned Stirling numbers of the first kind are defined by
$$
i!\binom{x+i-1}i=x(x+1)\cdots(x+i-1)=\sum_{j=1}^iS_1(i,j)x^j
$$
and satisfy the recursion
$$
S_1(i+1,j)=S_1(i,j-1)+i\,S_1(i,j)\quad\mbox{ for }1\le j\le i
$$
with $S_1(i,j)=0$ if $i<j$ or $j=0$.
Therefore we can write (\ref{eq:constants}) as
\begin{gather*}
\beta^\ell\left(\sum_{k=1}^\ell\frac{S_1(\ell,k)}{\ell!}q^k+
\sum_{k=1}^{\ell-1}\frac{S_1(\ell-1,k)}{(\ell-1)!}q^k+\cdots+q\right)
\qquad\qquad\qquad\qquad\qquad\qquad \\
= \sum_{j=1}^\ell\frac{S_1(\ell,j)}{\ell!}(\beta q+c_{\ell-1})^j+
\sum_{j=1}^{\ell-1}\frac{S_1(\ell-1,j)}{(\ell-1)!}(\beta q+c_{\ell-2})^j+\cdots+
\beta q+c_0, \\
\beta^\ell\sum_{i=1}^\ell\sum_{k=1}^i\frac{S_1(i,k)}{i!}q^k = 
\sum_{i=1}^\ell\sum_{j=1}^i\frac{S_1(i,j)}{i!}
\sum_{k=0}^j\binom jkc_{i-1}^{j-k}\beta^kq^k, \\
\beta^{\ell-k}\sum_{i=k}^\ell\frac{S_1(i,k)}{i!}=
\sum_{i=k}^\ell\sum_{j=k}^i\frac{S_1(i,j)\,j!}{i!\,(j-k)!\,k!}c_{i-1}^{j-k}\quad
\mbox{ for }0\le k\le\ell.
\end{gather*}
Since the last equation holds for $k=\ell$ and
$$
\beta^{\ell-k}\sum_{i=k}^\ell\frac{S_1(i,k)}{i!}=
\sum_{i=k}^\ell\frac{S_1(i,k)}{i!}+\frac{c_k}{k!}+
\sum_{i=k+2}^\ell\sum_{j=k+1}^i\frac{S_1(i,j)\,j!}{i!\,(j-k)!\,k!}c_{i-1}^{j-k}
$$
for $0\le k<\ell$ by the definition of $c_k$, the lemma is proved. 
\end{proof}

\begin{remark}\label{rem:ccc}
The formula for $c_k$ can be simplified using
$$
\sum_{i=k}^\ell\frac{S_1(i,k)}{i!}=\left\{\begin{array}{cl}S_1(\ell+1,k+1)/\ell!
& \mbox{for }k\ge1, \\ 0 & \mbox{for }k=0.\end{array}\right.
$$ 
Note that $c_{\ell-1}$ is the constant $c$ in the proof of 
Lemma~\ref{lem:long}, 
$$
c_{\ell-1}=(\beta-1)\frac{S_1(\ell+1,\ell)}\ell=
\frac{(\beta-1)(\ell+1)}2\quad\mbox{ for }\ell\ge2.
$$
Since $S_1(\ell+1,\ell-1)=S_1(\ell,\ell-2)+\ell\frac{\ell(\ell-1)}2=
\frac{(3\ell+2)(\ell+1)\ell(\ell-1)}{24}$, we have
\begin{align*}
c_{\ell-2} & = (\beta^2-1)\frac{(3\ell+2)(\ell+1)}{24}-\frac{\ell-1}2c_{\ell-1}-
\frac12c_{\ell-1}^2 \\
& = c_{\ell-1}\left(1-\frac{\beta+1}{12}\right) = \frac{(\beta-1)(\ell+1)}2-
\frac{(\beta^2-1)(\ell+1)}{24}\quad\mbox{ for }\ell\ge3.
\end{align*}
\end{remark}

We now turn to our main counting argument that will be used to obtain
that recognizability is not preserved through multiplication by a
constant $\lambda$. 
Recall that $f_\lambda:\mathcal{B}_\ell\to\mathcal{B}_\ell$ is defined 
by $f_\lambda(w)=\mathrm{rep}_\ell(\lambda\,\mathrm{val}_\ell(w))$.

\begin{lemma}\label{lem:notregular}
Let $A$ be a $k$-dimensional linear subset of $\mathbb{N}^\ell$ for some integer $k<\ell$ and 
$B=\Psi^{-1}(A)\cap\mathcal{B}_\ell$ be the corresponding subset of $\mathcal{B}_\ell$. 
If $\Psi(f_{\beta^\ell}(B))$ contains a sequence 
$x^{(n)}=(x_1^{(n)},\ldots,x_\ell^{(n)})$ such that 
$\min(x_{j_1}^{(n)},x_{j_2}^{(n)},\ldots,x_{j_{k+1}}^{(n)})\to\infty$ as 
$n\to\infty$ for some $j_1<j_2<\cdots<j_{k+1}$, then $f_{\beta^\ell}(B)$ is not 
regular.  
\end{lemma}

\begin{proof}
Since $A$ is a $k$-dimensional linear subset of $\mathbb N^\ell$, we clearly 
have 
$$
\#\{w\in B:|w|\le n\}=\#\{x\in A:x_1+\cdots +x_\ell \le n\}=\Theta(n^k)
$$
and, by Lemma \ref{lem:long}, 
$\#\{w\in f_{\beta^\ell}(B):|w|\le n\}=\Theta(n^k)$.
Thus $f_{\beta^\ell}(B)$ is regular if and only if $\Psi(f_{\beta^\ell}(B))$ is 
a finite union of at most $k$-dimensional sets as in Lemma \ref{lem:rec}.
Since the sequence $x^{(n)}$ cannot occur in such a finite union, 
$f_{\beta^\ell}(B)$ is not regular. 
\end{proof}

The coefficients $c_{\ell-1}$ and $c_{\ell-2}$ (explicitely given in
Remark~\ref{rem:ccc}) are rational numbers. In the next two
propositions, we discuss the fact that these coefficients could be
integers and we rule out all the possible cases.

\begin{proposition}\label{pro:c1}
If $\frac{(\beta-1)(\ell+1)}2\not\in\mathbb Z$ or 
$\frac{(\beta^2-1)(\ell+1)}{24}\not\in\mathbb Z$ (and $\ell\ge3,\,\beta\ge2$), then
$f_{\beta^\ell}(a_\ell^*)$ is not regular.
\end{proposition}

\begin{proof}
We use notation of the proof of Lemma~\ref{lem:long}.

{\bf First case} : $c_{\ell-1}=\frac{(\beta-1)(\ell+1)}2\not\in\mathbb Z$ \\ 
We have $z_\ell=\beta q+c_{\ell-1}+1/2$, 
$z_{\ell-1}=2^{-1/(\ell-1)}\beta q+o(q)$, hence
\begin{align*}
|f_{\beta^\ell}(a_\ell^q)|_{a_1} & =
(1-2^{-1/(\ell-1)})\beta q+o(q), \\
\sum_{j=2}^\ell| f_{\beta^\ell}(a_\ell^q)|_{a_j}
& = 2^{-1/(\ell-1)}\beta q+o(q),
\end{align*}
and $f_{\beta^\ell}(a_\ell^*)$ is not regular 
by Lemma~\ref{lem:notregular}. 

{\bf Second case} : $c_{\ell-1}=\frac{(\beta-1)(\ell+1)}2\in\mathbb Z$\\
We have $z_\ell=\beta q+c_{\ell-1}$, 
$z_{\ell-1}=\beta q+\mathcal O(1)$ and $z_{\ell-2}=d\beta q+o(q)$ with 
$0\le d\le 1$.
By comparing the coefficients of $q^{\ell-2}$, we obtain
$$
z_{\ell-1}=\beta q+c_{\ell-2}+1-d^{\ell-2}
$$
Since in this case $c_{\ell-2}=\frac{(\beta-1)(\ell+1)}2-
\frac{(\beta^2-1)(\ell+1)}{24}\not\in\mathbb Z$, we have $0<d<1$, hence
$$
|f_{\beta^\ell}(a_\ell^q)|_{a_2}=
(1-d)\beta q+o(q),\quad
\sum_{j=3}^\ell|f_{\beta^\ell}(a_\ell^q)|_{a_j}
=d\beta q+o(q),
$$
and $f_{\beta^\ell}(a_\ell^*)$ is not regular 
by Lemma~\ref{lem:notregular}.
\end{proof}

\begin{proposition}\label{pro:c2}
If $\frac{(\beta-1)(\ell+1)}2\in\mathbb Z$ and $\frac{(\beta^2-1)(\ell+1)}{24}\in\mathbb Z$ (and $\ell\ge3,\,\beta\ge2$), then 
$f_{\beta^\ell}(a_1^*a_\ell^*)$ is not regular.
\end{proposition}

\begin{proof}
If we choose $q$ large enough with respect to $p$, e.g. $q=p^3$, then we have
\begin{multline*}
\beta^\ell\left(\binom{p+q+\ell-1}{\ell}+\binom{q+\ell-2}{\ell-1}+
\binom{q+\ell-3}{\ell-2}+\cdots+\binom q1\right) \\
= \binom{\beta(p+q)+c_{\ell-1}+\ell-1}{\ell}+
\binom{\beta q-(\beta-1)\beta p+c_{\ell-2}+\ell-2}{\ell-1} \\
+\binom{\beta q-\frac{(\beta-1)\beta}2(\beta p)^2+\mathcal O(p)}{\ell-2}+
\mathcal O\big(q^{\ell-3}\big).
\end{multline*}
Indeed, this equation holds for $p=0$ by Lemma~\ref{lem:constants}.
Therefore the coefficients of $q^\ell p^0$, $q^{\ell-1}p^0$ and $q^{\ell-2}p^0$
on the left-hand side are equal to those on the right-hand side.
It is easy to see that the same holds for $q^{\ell-1}p^1$, $q^{\ell-2}p^2$ and 
$q^{\ell-3}p^3$.
For $q^{\ell-2}p^1$ and $q^{\ell-3}p^2$, consider the following equations:
\begin{align*}
(\ell-2)!\,\beta^{1-\ell}\big[q^{\ell-2}p^1\big]:\quad & 
\beta\frac{\ell-1}2=c_{\ell-1}+\frac{\ell-1}2-(\beta-1), \\
(\ell-3)!\,\beta^{1-\ell}\big[q^{\ell-3}p^2\big]:\quad &
\beta\frac{\ell-1}4=\frac{c_{\ell-1}}2+\frac{\ell-1}4+\frac{(\beta-1)^2}2-
\frac{(\beta-1)\beta}2. 
\end{align*}
If the $\mathcal O(p)$ term is chosen properly, then the coefficient of 
$q^{\ell-3}p^1$ vanishes as well and $\mathcal O\big(q^{\ell-3}\big)$ remains.
Since $c_\ell,c_{\ell-1}\in\mathbb Z$, we have thus 
\begin{gather*}
|f_{\beta^\ell}(a_1^pa_\ell^q)|_{a_1}=
\beta^2p+\mathcal O(1), \\
|f_{\beta^\ell}(a_1^pa_\ell^q)|_{a_2}=
\frac{(\beta-1)\beta^3}2p^2+\mathcal O(p), \\ 
\sum_{j=3}^\ell |f_{\beta^\ell}(a_1^pa_\ell^q)|_{a_j}
=\beta q+\mathcal O(p^2),
\end{gather*}
and $f_{\beta^\ell}(a_1^*a_\ell^*))$ is not regular 
by Lemma~\ref{lem:notregular}.
\end{proof}

\begin{example}
    We just illustrate some of the above computations. If $\ell=3$,
    then we have $c_2=2(\beta-1)$, $c_1=2(\beta-1)-(\beta^2-1)/6$ and
$$
c_0=-\frac{c_1}2-\frac{c_1^2}2-\frac{c_2}3-\frac{c_2^2}2-\frac{c_2^3}6=
-\frac{(\beta^2-1)^2}{72}-(\beta^3-1)-\frac{\beta^2-1}4+2(\beta-1).
$$
If $\beta\equiv\pm1\pmod 6$, then this gives
$$
f_{\beta^3}(a_3^q)=a_1^{\frac{\beta^2-1}6}
a_2^{\frac{(\beta^2-1)^2}{72}+\beta^3-1+\frac{\beta^2-1}{12}}
a_3^{\beta q-\frac{(\beta^2-1)^2}{72}-(\beta^3-1)-\frac{\beta^2-1}4+2(\beta-1)}.
$$
In particular, this latter formula shows that $a_3^*$ cannot be
used to prove that multiplication by $\beta^3$ does not preserve
recognizability when $\beta\equiv\pm1\pmod 6$. Thanks to
Proposition~\ref{pro:c1}, $f_{\beta^3}(a_3^q)$ is regular if and
only if $\beta\equiv\pm1\pmod 6$.

Otherwise, i.e., if $1-\beta^2\equiv j\pmod 6$ with $j\in\{1,3,4\}$, then $z_3=\beta q+c_2$,
$z_2=\beta q+c_1+1-j/6$ and
$$
z_1=\frac j6\beta q+c_0-\frac{(1-j/6)^2}2-(1-j/6)c_1-\frac{1-j/6}2.
$$
\end{example}

If we collect results from Theorems \ref{the:6}, \ref{the:slender},
\ref{the:LR} and Propositions \ref{pro:c1} and \ref{pro:c2}, we obtain
the main result about multiplication by a constant.

\begin{theorem}
Let $\ell,\,\lambda$ be positive integers. 
For the abstract numeration system
    $$S=(a_1^*\cdots a_\ell^*,\{a_1<\cdots<a_\ell\}),$$ 
    multiplication by $\lambda\ge 2$ preserves
    $S$-recognizability if and only if one of the following condition is satisfied :
    \begin{itemize}
        \item $\ell=1$
        \item $\ell=2$ and $\lambda$ is an odd square.
    \end{itemize}
\end{theorem}

\begin{proof}
    The case $\ell=1$ is ruled out by Theorem \ref{the:slender}, the
    case $\ell=2$ is given by Theorem \ref{the:LR}. Consider $\ell\ge
    3$. Thanks to Theorem \ref{the:6}, it suffices to consider
    $\lambda$ of the $\beta^\ell$ and the conclusion follows from
    Propositions \ref{pro:c1} and \ref{pro:c2}.
\end{proof}

 \vskip 30pt

\section*{\normalsize 6. Structural properties of $\mathcal{B}_\ell$ seen through $f_{\beta_\ell}$}

In this independent section, we inspect closely how a word is
transformed when applying $f_{\beta^\ell}$. To that end,
$\mathcal{B}_\ell$ (or equivalently $\mathbb{N}$) is partitioned into
regions where $f_{\beta^\ell}$ acts differently. Thanks to our
discussion, we are able to detect some kind of pattern occurring
periodically within these regions. To have a flavor of the
computations involved in this section, the reader could first have a
look at Example~\ref{ex:par}. 
According to Lemma \ref{lem:long}, we define a partition of 
$\mathbb{N}$.

\begin{definition}
    For all $i\in\{0,1,\ldots,\beta\}$ and $k\in\mathbb{N}$ large 
    enough, we define
    $$\mathcal{R}_{i,k}:=\left\{n\in\mathbb{N}: |\rep{\ell}(n)|=k
    \text{ and }|\rep{\ell}(\beta^\ell n)|=\beta\,
    k+\left\lceil\frac{(\beta-1)(\ell+1)}{2}\right\rceil-i\right\}.$$
\end{definition}

\begin{lemma}\label{lem:congrumod}
    If $\beta=\prod_{i=1}^k p_i^{\theta_i}$ where $p_1,\ldots,p_k$ are
    prime numbers greater than $\ell$ and the $\theta_i$'s are 
    positive integers, then for any $u\ge\ell$, we have
    $$\binom{u}{\ell}\equiv\binom{u+\beta^\ell}{\ell}
    \pmod{\beta^\ell}.$$
\end{lemma}

\begin{proof}
    Let $u,v\ge\ell$. One has
    $$\binom{v}{\ell}-\binom{u}{\ell}
    =\frac{v(v-1)\cdots(v-\ell+1)-u(u-1)\cdots(u-\ell+1)}{\ell!}.$$
    The numerator on the r.h.s. is an integer divisible by $\ell!$.
    Moreover, this numerator is also clearly divisible by $v-u$
    (indeed, it is of the form $P(v)-P(u)$ for some polynomial $P$).\\
    Notice that for $v=u+\beta^\ell$, the corresponding numerator is
    divisible by $\ell!$ and also by $\beta^\ell$. But since any prime
    factor of $\beta$ is larger than $\ell$, $\ell!$ and $\beta^\ell$
    are relatively prime. Consequently, the corresponding numerator is
    divisible by $\beta^\ell \ell!$.
\end{proof}

An inspection of multiplication by
$\beta^\ell$ using the partition induced by Lemma~\ref{lem:long}
provides us with the following observation.

\begin{proposition}\label{pro:min}
    Let $m_{i,k}=\min\mathcal{R}_{i,k}$ for $k\ge 0$ and 
    $i\in\{0,\ldots,\beta\}$.
    If $\beta$ satisfies the condition of Lemma~\ref{lem:congrumod}, 
    then 
    $$
    |\rep{\ell}(\beta^\ell m_{i,k})|_{a_j}=
    |\rep{\ell}(\beta^\ell m_{i,k+\beta^{\ell-1}})|_{a_j}
    $$
    for all $k$ large enough and $j\in\{2,\ldots,\ell\}$.
    Furthermore, 
    $$
    |\rep{\ell}(\beta^\ell m_{i,k+\beta^{\ell-1}})|_{a_1}
    =|\rep{\ell}(\beta^\ell m_{i,k})|_{a_1}+\beta^\ell.
    $$
    If $i<\beta$, then $m_{i,k}=\lceil C_i(k)/\beta^\ell\rceil$ with
    $$
    C_i(k)=
\mathrm{val}_\ell\left(a_1^{\beta\,k+\frac{(\beta-1)(\ell+1)}2-i}\right)
    =\binom{\beta\,k+\frac{(\beta-1)(\ell+1)}{2}-i+\ell-1}{\ell}.  
    $$
\end{proposition}

\begin{proof}
    For $i=\beta$, we clearly have 
    $m_{\beta,k}=\mathrm{val}_\ell(a_1^k)$ if $\mathcal R_{\beta,k}$ is 
    non-empty, and it is easily verified that $\mathcal R_{\beta,k}$ is
    non-empty if $k$ is large enough (and $\ell\ge2$). \\
    For $i<\beta$, note first that $(\beta-1)(\ell+1)$ is even since 
    $\beta$ satisfies the condition of Lemma~\ref{lem:congrumod}.
    Thus we have
    $$
    C_i(k)\le \beta^\ell m_{i,k} < C_{i-1}(k)
    $$
    Since $m_{i,k}-1\in\mathcal{R}_{i+1,k}$, we also obtain
    $$
    C_{i+1}(k)+\beta^\ell\le \beta^\ell m_{i,k}<C_i(k)+\beta^\ell.
    $$
Therefore $m_{i,k}=\lceil C_i(k)/\beta^\ell\rceil$ and there exists a 
unique integer $\mu_i(k)$ such that
$$
\beta^\ell m_{i,k} = C_i(k)+\mu_i(k) \quad\text{ and }\quad 0\le 
\mu_i(k)<\beta^\ell.
$$
In particular, there exists also a unique integer 
$\mu_i(k+\beta^{\ell-1})$ such that 
$$
\beta^\ell m_{i,k+\beta^{\ell-1}} = C_i(k+\beta^{\ell-1})+
\mu_i(k+\beta^{\ell-1})
\quad\text{ and }\quad 0\le \mu_i(k+\beta^{\ell-1})<\beta^\ell.$$
From Lemma~\ref{lem:congrumod}, we deduce that $C_i(k)\equiv
C_i(k+\beta^{\ell-1})\pmod{\beta^\ell}$ and consequently,
$\mu_i(k)=\mu_i(k+\beta^{\ell-1})$. 
From Lemma~\ref{lem:val}, we deduce that
$$
\mathrm{rep}_{\ell}(\beta^\ell m_{i,k})=a_1^{t}\, 
\rep{\{a_2,\ldots,a_\ell\}}(\mu_i(k)),
$$
where $t$ is such that 
$|\rep{\ell}(\beta^\ell m_{i,k})|=\beta\,k+\frac{(\beta-1)(\ell+1)}2-i$,
and
$$
\hskip42mm
\rep{\ell}(\beta^\ell m_{i,k+\beta^{\ell-1}})=a_1^{t+\beta^\ell}\, 
\rep{\{a_2,\ldots,a_\ell\}}(\mu_i(k)).\hskip37mm \qedhere
$$
\end{proof}

\begin{remark}
    In the previous proposition, we were interested in the first word
    in $\mathcal{R}_{i,k}$ but we can even describe how multiplication
    by $\beta^\ell$ affects representations inside
    $\mathcal{R}_{i,k}$.  With notation of the previous proof, for any
    $n\in\mathcal R_{i,k}$ (and $k$ large enough), we have
$$
\mathrm{rep}_\ell(\beta^\ell n)=a_1^t\,
\rep{\{a_2,\ldots,a_\ell\}}(\mu_i(k)+\beta^\ell(n-m_{i,k}))
$$ 
with $t$ such that 
$|\rep{\ell}(\beta^\ell n)|=\beta\,k+\frac{(\beta-1)(\ell+1)}2-i$.
\end{remark}

\begin{example}\label{ex:par}
    Let $\ell=3$ and $\beta=5$. The number $171717$ (resp. $172739$)
    is the first element belonging to $\mathcal{R}_{4,100}$ (resp.
    $\mathcal{R}_{3,100}$). We have
    $$\rep{3}(171717)=a^{95}b^3c^2 \text{ and }\rep{3}(5^3\,
    171717)=a^{490}\underline{b^{14}c^0},$$
    $$\rep{3}(172739)=a^{55}b^{41}c^4 \text{ and }\rep{3}(5^3\,
    172739)=a^{493}\underline{b^{0}c^{12}}.$$
    Therefore $\mu_4(100)=\val{\{b,c\}}(b^{14})=105$ (resp.
    $\mu_3(100)=\val{\{b,c\}}(c^{12})=90$).  
    The number $333396$ (resp.  $334986$) is the smallest element in
    $\mathcal{R}_{4,125}$ (resp. $\mathcal{R}_{3,125}$),
    $$\rep{3}(333396)=a^{119}b^6c^0\text{ and }\rep{3}(5^3\,
    333396)=a^{615}\underline{b^{14}c^0},$$ 
    $$\rep{3}(334986)=a^{69}b^{41}c^{15}\text{ and }\rep{3}(5^3\,
    334986)=a^{618}\underline{b^{0}c^{12}}.$$
    We have $\#\mathcal{R}_{4,100}=1022$,
    $\#\mathcal{R}_{4,125}=1590$ and get the following table.
$$\begin{array}{|c||c|c|c|}
\hline
j & \Psi(\rep{3}(5^3(m_{4,100}+j)))&  \Psi(\rep{3}(5^3 
(m_{4,125}+j))) & \Psi(\rep{\{b,c\}}(\mu_4(100)+5^3 j))\\
\hline
0&(490,14,0) &    (615,14,0) & (14,0)\\
1&(484,0,20) &    (609,0,20)& (0,20) \\
2&(478,22,4) &    (603,22,4)& (22,4)\\
\vdots &\vdots &  \vdots &  \vdots\\
1021&(0,34,470) & (125,34,470)& (34,470) \\
1022& \times   & (124,415,90) & (415,90)\\
\vdots& \vdots & \vdots &  \vdots\\
1589 & \times  & (0,34,595) &(34,595)\\
\hline
\end{array}$$
\end{example}

 \vskip 30pt

\section*{\normalsize Acknowledgments} 
We thank P.~Lecomte for fruitful discussions during the elaboration of
this paper.

\end{document}